\documentclass{amsart}
\usepackage{amsmath,amsthm,amsfonts, amssymb,amscd}
\usepackage[all]{xy}

\newtheorem{theorem}{Theorem}[section]
\newtheorem{example}[theorem]{Example}
\newtheorem{corollary}[theorem]{Corollary}
\newtheorem{proposition}[theorem]{Proposition}

\newcommand{\RDP}{\mbox{\rm RDP}}
\newcommand{\RIP}{\mbox{\rm RIP}}
\newcommand{\scom}{\stackrel{\mbox{\rm c}}{\longleftrightarrow}}

\begin{document}
\title[Representable Effect Algebras and Observables]{Representable Effect Algebras and Observables}
\author[Anatolij Dvure\v censkij]{Anatolij Dvure\v censkij$^{1,2}$}
\date{}
\maketitle

\begin{center}  \footnote{


The paper has been supported by the Slovak Research and Development Agency under the contract APVV-0178-11, the grant VEGA No. 2/0059/12 SAV and by
CZ.1.07/2.3.00/20.0051. }
Mathematical Institute,  Slovak Academy of Sciences,\\
\v Stef\'anikova 49, SK-814 73 Bratislava, Slovakia\\
$^2$ Depart. Algebra  Geom.,  Palack\'{y} Univer.\\
CZ-771 46 Olomouc, Czech Republic\\

E-mail: {\tt
dvurecen@mat.savba.sk}
\end{center}

\begin{abstract}
We introduce a class of monotone $\sigma$-complete effect algebras, called representable, which are  $\sigma$-homomorphic images of a class of monotone $\sigma$-complete effect algebras of functions taking values in the interval $[0,1]$ and with pointwise defined effect algebra operations. We exhibit different types of compatibilities and show their connection to representability. Finally, we study observables and show situations when information of an observable on all intervals of the form $(-\infty,t)$ gives full information about the observable.
\end{abstract}

\section{Introduction}

Effect algebras form an important family of quantum structures introduced by Foulis and Bennett \cite{FoBe} which model quantum mechanical events. They are partial algebras with the primary notion of addition, $+,$  such that
$a + b$ describes the disjunction of two mutually excluding events
$a$ and $b.$ A basic model of effect algebras is the system $\mathcal E(H)$ of Hermitian operators of a (real, complex or quaternionic) Hilbert space $H$ that are between the zero and the identity operator. Effect algebras generalize many quantum structures like Boolean algebras, orthomodular lattices and posets, orthoalgebras, and MV-algebras. Effect algebras combine in an algebraic way both sharp and fuzzy features of a measurement process in quantum mechanics, and during the last two decades they become intensively studied by many authors. For a comprehensive source of information about effect algebras we recommend \cite{DvPu}.

A basic notion of theory effect algebras is compatibility. It means that two elements $a$ and $b$ can be split into three mutually excluding elements $a_1,b_1,c$ such that $a=a_1+c,$ $b=b_1+c$ and $a_1+b_2+c$ is defined in the effect algebra. There are different types of compatibilities: strong, internal, f-compatibility etc, see \cite{Dvu, DvPu}. The Riesz Decomposition Property (RDP for short) says roughly speaking that any two decompositions of the same element have a common refinement. The orthodox effect algebra of Hilbert space quantum mechanics $\mathcal E(H)$ does not fulfill RDP, but it can be covered by monotone $\sigma$-complete effect subalgebras with RDP, see e.g. \cite{Pul1}.  A more general notion than effect algebras with RDP are homogeneous algebras introduced in \cite{Jen}.

A fundamental result by Loomis and Sikorski, the Loomis-Sikorski Theorem, see e.g. \cite{Sik}, says that every Boolean $\sigma$-algebra is a $\sigma$-homomorphic image of a $\sigma$-algebra of subsets. It was generalized also for $\sigma$-complete MV-algebras \cite{Mun, LoomD} and for monotone $\sigma$-complete effect algebras with the Riesz Decomposition Property, \cite{BCD}, showing that every such an algebra is a $\sigma$-homomorphic image of a special family of $[0,1]$-valued functions, called a tribe and an effect-tribe, respectively, where the basic corresponding algebraic operations among functions are defined by points.

The aim of the present paper is the study of relation between strong compatibility or internal compatibility and possibility to cover a monotone $\sigma$-complete effect algebra by a system of $\sigma$-complete MV-algebras or by a system of effect algebras with RDP; such effect algebras are said to be weakly representable. These questions were studied also in \cite{CIJTP, Jen, Pul1, Pul2, NiPa}. The first such a kind of results was established in \cite{Rie} where it was shown that every lattice ordered effect algebra can be covered by blocks, and every such a block is an MV-algebra.

An observable models a quantum mechanical measurement. Mathematically it is a $\sigma$-homomorphism from the Borel $\sigma$-algebra into a monotone $\sigma$-complete effect algebra. Therefore, we study also a question when an observable can be uniquely determined knowing only  its values for all intervals of the form $(-\infty,t).$ We extend known results from \cite{DvKu}.

The paper is organized as follows. Section 2 gathers the basic facts on effect algebras and MV-algebras. Different types of compatibilities are studied in Section 3. We present also a result when a monotone $\sigma$-complete effect algebra can be covered by a system of $\sigma$-complete MV-algebras. More details on representable or weakly representable effect algebras will be done in Section 4. Finally, Section 5 shows situations when a partial information on an observable will imply the whole information on the observable.

\section{Basic Facts on Effect Algebras}

According to \cite{FoBe}, an {\it effect algebra} is  a partial algebra $E =
(E;+,0,1)$ with a partially defined operation $+$ and with two constant
elements $0$ and $1$  such that, for all $a,b,c \in E$,
\begin{enumerate}

\item[(i)] $a+b$ is defined in $E$ if and only if $b+a$ is defined, and in
such a case $a+b = b+a;$

 \item[(ii)] $a+b$ and $(a+b)+c$ are defined if and
only if $b+c$ and $a+(b+c)$ are defined, and in such a case $(a+b)+c
= a+(b+c);$

 \item[(iii)] for any $a \in E$, there exists a unique
element $a' \in E$ such that $a+a'=1;$

 \item[(iv)] if $a+1$ is defined in $E$, then $a=0.$
\end{enumerate}

If we define $a \le b$ if and only if there exists an element $c \in
E$ such that $a+c = b$, then $\le$ is a partial ordering on $E$, and
we write $c:=b-a.$ It is clear that $a' = 1 - a$ for any $a \in E.$ For more information about effect algebras, see \cite{DvPu}.  A {\it homomorphism}  from an effect algebra $E_1$ into another one $E_2$ is any mapping $h: E_1 \to E_2$ such that  (i) $h(1)=1$ (ii) if $a+b$ is defined in $E_1$ so is defined $h(a)+h(b)$ in $E_2$ and $h(a+b)= h(a)+h(b).$ A subset $F$ of an effect algebra $E$ is an {\it effect subalgebra} of $E$ if (i) $0,1 \in F,$ (ii) if $a\in F,$ then $a'\in F,$ and (iii) if $a,b\in F$ and $a+b \in E,$ then $a+b \in F.$

Two important families of effect algebras are the following.
(1) Let $E$ be a system of fuzzy sets on $\Omega,$ that is $E \subseteq [0,1]^\Omega,$ such that
(i) $1 \in E$, (ii) $f \in E$ implies $1-f \in E$, and (iii) if $f,g
\in E$ and $f(\omega) \le 1 -g(\omega)$ for any $\omega \in \Omega$,
then $f+g \in E$. Then $E$ is an effect algebra of fuzzy sets which
is not necessarily a Boolean algebra.

(2) If $G$ is an Abelian partially ordered group written additively, $u \in G^+:=\{g \in G \colon g \ge 0\}$, set $\Gamma(G,u):=[0,u]=\{g \in G: 0 \le g \le u\}.$ Then $(\Gamma(G,u);+,0,u)$ is an effect algebra where $+$ is the group addition of elements from $\Gamma(G,u)$  if it exists in $\Gamma(G,u).$ In particular, if $G =\mathbb R,$ the group of real numbers, then $[0,1]=\Gamma(\mathbb R,1)$ is the standard effect algebra of the real interval $[0,1].$ For more information on Abelian partially ordered groups see \cite{Goo}.

We say that an effect algebra $E$ satisfies the Riesz Decomposition Property (RDP for short) if for all $a_1,a_2,b_1,b_2 \in E$ such that $a_1 + a_2 = b_1+b_2,$ there are four elements $c_{11},c_{12},c_{21},c_{22}$ such that $a_1 = c_{11}+c_{12},$ $a_2= c_{21}+c_{22},$ $b_1= c_{11} + c_{21}$ and $b_2= c_{12}+c_{22}.$

We define $\sum_{i=1}^n a_i:= a_1+\cdots +a_n$, if the element on the right-hand exists in $E.$ A system of elements $(a_i: i \in I)$ is said to be {\it summable} if, for any finite set $F$ of $I,$ the element $a_F:= \sum_{i\in F} a_i$ is defined in $E.$ We define $\sum_{i \in \emptyset} a_i:=0.$  If there is an element $a:= \sup \{a_F:  F$ is a finite subset of $I\},$ we call it the {\it sum} of $\{a_i: i \in I\}$ and we write $a = \sum_{i \in I}a_i.$

An effect algebra $E$ is {\it monotone} $\sigma$-{\it complete} if, for any sequence $a_1 \le a_2\le \cdots,$ the element $a = \bigvee_n a_n$  is defined in $E$ (we write $\{a_n\}\nearrow a$).

If $E$ and $F$ are two monotone $\sigma$-complete effect algebras, a homomorphism $h:E \to F$ is said to be a $\sigma$-{\it homomorphism} if $\{a_n\} \nearrow a$ implies $\{h(a_n)\} \nearrow h(a),$ for $a, a_1,\ldots \in E.$

A {\it state} on an effect algebra is any mapping $s: E \to [0,1]$ such that $s(1)=1$ and $s(a+b)=s(a)+s(b)$ whenever $a+b$ is defined in $E.$ We note that it can happen that an effect algebra is stateless. However, if $E$ has RDP, then it possesses at least one state. A state $s$ is $\sigma$-additive if $s(a) = \lim_n s(a_n)$ whenever $\{a_n\}\nearrow a.$

An important family of effect algebras is the family of
MV-algebras introduced by Chang \cite{Cha}.

An MV-algebra is an algebra $M = (M;\oplus, ^*,0,1)$
of type $\langle 2,1,0,0\rangle$ such that, for all $a,b,c \in M$, we have

\begin{enumerate}
\item[(i)]  $a \oplus  b = b \oplus a$;
\item[(ii)] $(a\oplus b)\oplus c = a \oplus (b \oplus c)$;
\item[(iii)] $a\oplus 0 = a;$
\item[(iv)] $a\oplus 1= 1;$
\item[(v)] $(a^*)^* = a;$
\item[(vi)] $a\oplus a^* =1;$
\item[(vii)] $0^* = 1;$
\item[(viii)] $(a^*\oplus b)^*\oplus b=(a\oplus b^*)^*\oplus a.$
\end{enumerate}

If we define a partial operation $+$ on $M$ in such a way that $a+b$
is defined in $M$ if and only if $a \le b^*$ and  we set
$a+b:=a\oplus b$, then $(M;+,0,1)$ is an effect algebra with RDP.
For example, if $G$ is an Abelian lattice ordered group and $u \ge0,$ then
$(\Gamma(G,u); \oplus, ^*, 0,u)$, where  $\Gamma(G,u):= \{g\in G: 0\le g \le u\},$ $a\oplus b :=(a+b) \wedge u$, and $a^* := u-a$ $(a,b \in \Gamma(G,u))$ is an MV algebra, and every MV algebra arises in this way. Every MV-algebra is a distributive lattice. If it is also a $\sigma$-lattice, we call it a {\it $\sigma$-complete MV-algebra} ($\sigma$-MV-algebra for short).

The following two classes of $\sigma$-complete MV-algebras and monotone $\sigma$-complete effect algebras are those of $[0,1]$-valued functions with pointwise defined algebraic operations. They play an analogous role as $\sigma$-algebras do for $\sigma$-complete Boolean algebras in the Loomis-Sikorski Theorem, see \cite{Mun, LoomD} and \cite{BCD}.

(1) A {\it tribe} on $\Omega \ne \emptyset$
is a collection ${\mathcal T}$ of functions from $[0,1]^\Omega$ such
that (i) $1 \in {\mathcal T}$, (ii) if $f \in {\mathcal T}$, then $1 - f \in
{\mathcal T},$ and (iii) if $\{f_n\}$ is a sequence from ${\mathcal T}$,
then $\min \{\sum_{n=1}^\infty f_n,1 \}\in {\mathcal T}.$  A tribe is
always a $\sigma$-complete MV-algebra. For example if $f_n = \chi_{A_n},$ where $\chi_A$ is the characteristic function of the set $A,$  then  $\min \{\sum_{n=1}^\infty \chi_{A_n},1 \} = \chi_{\bigcup_n A_n}.$

(2) An {\it effect-tribe}  is any system ${\mathcal T}$ of $[0,1]$-valued functions on
$\Omega\ne \emptyset $ such that (i) $1 \in {\mathcal T}$, (ii) if $f
\in {\mathcal T},$ then $1-f \in {\mathcal T}$, (iii) if $f,g \in {\mathcal T}$,
$f \le 1-g$, then $f+g \in {\mathcal T},$ and (iv) for any sequence
$\{f_n\}$ of elements of ${\mathcal T}$ such that $f_n \nearrow f$
(pointwise), then $f \in {\mathcal T}$. Every
effect-tribe is a monotone $\sigma$-complete effect algebra.

We remind also the definition of an {\it orthomodular poset}, \cite{DvPu}: It is a poset $L$ with the least and greatest elements $0$ and $1,$ equipped with a mapping $a \mapsto a^\bot$ such that for all $a,b\in L$, we have (i) $a^{\bot\bot}=a$ (ii) $b^\bot \le a^\bot$ if $a \le b,$ (iii) if $a \le b^\bot$, $a\vee b \in L,$ (iv) $a \vee a^\bot = 1,$ and (v) if $a\le b$ implies $b = a\vee (a \vee b^\bot)^\bot$ $ (=  a \vee (b \wedge a^\bot)).$ If we set $a+b=a\vee b$ whenever $a \le b^\bot,$ then $(L;+,0,1)$ is an effect algebra. An orthomodular poset $L$ is  $\sigma$-{\it orthocomplete} if, for any sequence $\{a_n\}$ such that $a_n \le a_m^\bot$ for $n\ne m,$  $\bigvee_n a_n$ exists in $L.$ Then $L$ is also a monotone $\sigma$-complete effect algebra.

Finally we present one construction of effect algebras using a given family of effect algebras.
Let $((E_t;+_t, 0,1)\colon t \in T)$ be a system of effect algebras such that $E_s\cap E_t=\{0,1\}$ for all $s,t \in T,$ $s\ne t.$ Then $E=\bigcup_{t \in T}E_t,$ with a partially defined operation $+$ defined by $a+b=a +_t b$ iff $a,b \in E_t$ and $a+_t b$ exists in $E_t$, is an effect algebra called the {\it horizontal sum} of $(E_t\colon t \in T).$

\section{Compatibilities}

In this section, we define some types of compatible elements or subsets, respectively. We show that the strong compatibility may implies that a monotone $\sigma$-complete effect algebra can be covered by a system of blocks, and every such a block will be a $\sigma$-MV-algebra or an effect subalgebra with RDP, respectively.

We say that an effect algebra $E$ satisfies (i) the
{\it Riesz Interpolation Property} (RIP for short) if, for all
$x_1,x_2,y_1,y_2$ in $E$,  $x_i \le y_j$ for all $i,j$ implies
there exists an element $z \in E$ such that $x_i \le y_j$ for all
$i,j,$ (ii) {\it countable Riesz Interpolation Property} ($\sigma$-RIP for short) if, for any two sequences $\{x_i\}$ and $\{y_j\}$ of elements of $E$ with $x_i \le y_j$ for all $i,j$, there is an element $z \in E$ such that $x_i\le y_j$ for all $i,j.$

We remind that (1) if $E$ is a lattice, then $E$ has trivially
\RIP; the converse is not true as wee see below (Example \ref{ex:3.3}). (2) $E$ satisfies
RDP iff $a \le b+c,$ then there are two elements $b_1,c_1\in E$ such that $a=b_1+c_1,$
\cite[Lem 1.7.5]{DvPu}. (3) RDP implies
RIP, but the converse is not true (e.g. if $E = \mathcal L(H)$, then $E$
is a complete lattice but RDP fails)). (4) For monotone $\sigma$-complete effect algebras RIP and $\sigma$-RIP are equivalent, see \cite[Prop 4.1]{Pul2}.

A more general notion than an effect algebra with RDP is the following notion \cite{Jen}: We say that an effect algebra $E$ is {\it homogeneous} if, whenever $a,b,c \in E$ are such that $a \le b+c$, $a\le (b+c)',$ there are $a_1,a_2 \in E$ such that $a_1\le b,$ $a_2 \le c$ and $a=a_1\oplus a_2.$ We notice that (1) every effect algebra with RDP is homogeneous, (2) every lattice effect algebra is homogeneous, (3) every homogeneous effect algebra can be covered by a system of effect subalgebras $(E_t\colon t \in T)$ of $E$ such that every $E_t$ satisfies RDP, \cite[Thm 3.1, Cor 3.13]{Jen}.

Two elements $a$ and $b$ of an effect algebra $E$ are said to be
(i) {\it compatible} and write $a \leftrightarrow b$ if there
exist three elements $a_1, b_1, c \in E$ such that $a= a_1 + c,$
$b= b_1 + c$ and $a_1 + b_1 + c \in E$, and (ii) {\it strongly
compatible} and we write $a\stackrel{\mbox{\rm
c}}{\longleftrightarrow} b$ if there are three elements $a_1, b_1,
c \in E$ such that $a = a_1 + c,$ $b = b_1 + c$, $a_1 \wedge  b_1
= 0$ and $a_1 +b _1 + c \in E.$

We recall that the basic facts on these compatibilities were proved in \cite[Section 2]{CIJTP}:

\begin{proposition}\label{pr:3.1}
Let $E$ be an effect algebras.

\begin{enumerate}
\item[{\rm (i)}] If $a\stackrel{\mbox{\rm
c}}{\longleftrightarrow} b$, then $a \leftrightarrow b.$

\item[{\rm (ii)}] $a\leftrightarrow b$ $ (a\stackrel{\mbox{\rm
c}}{\longleftrightarrow} b)$ implies $b \leftrightarrow a$
$(b\stackrel{\mbox{\rm c}}{\longleftrightarrow} a).$

\item[{\rm (iii)}] $0
\stackrel{\mbox{\rm c}}{\longleftrightarrow} a \stackrel{\mbox{\rm
c}}{\longleftrightarrow} 1.$

\item[{\rm (iv)}] If $a \le b$, then
$a\stackrel{\mbox{\rm c}}{\longleftrightarrow} b$, $(b = (b-a) +
a, $ $a= 0 +a)$.

\item[{\rm (v)}] If $E$ satisfies \RDP, then $a \longleftrightarrow b$ for all $a,b\in E.$
\end{enumerate}

In addition, let $E$ satisfy \RIP.

\begin{enumerate}
\item[{\rm (vi)}] If $a = a_1 + c$, $b = b_1 +c$ with $a_1
\wedge b_1 = 0$ and $a_1 + b_1 +c \in E$, then $a\wedge b = c$, $a
\vee b = a_1 + b_1 +c.$

\item[{\rm (vii)}] $a \scom b$ if and only if $a \leftrightarrow b$ and $a\wedge b \in E$ if and only if $a \longleftrightarrow b$ and $a\vee b \in E.$

\item[{\rm (viii)}] $a \vee b, a\wedge b \in E$ and $(a \vee b) - b = a
-(a\wedge b).$

\item[{\rm (ix)}]    $a \vee b, a\wedge b \in E$ and $(a \vee b) - a = b
-(a\wedge b).$
\end{enumerate}
\end{proposition}

\begin{proposition}\label{pr:3.2}
Let a monotone $\sigma$-complete effect algebra $E$
satisfy {\rm RIP} and let, for a sequence $\{a_n\},$ we have $a_1\vee \cdots \vee a_n$ is defined in $E$ for every $n \ge 1.$ If $b \stackrel{\mbox{\rm
c}}{\longleftrightarrow} a_n$ for every $n \ge 1,$ then $a:=\bigvee_n a_n \in E,$  $b
\stackrel{\mbox{\rm c}}{\longleftrightarrow} a,$ and
$$
b \wedge \bigvee_n a_n = \bigvee_n (b\wedge a_n).
$$
\end{proposition}

\begin{proof}
Since $E$ is monotone $\sigma$-complete, the element $a = \bigvee_n a_n$ is defined  in $E.$ By (vi) of Proposition
\ref{pr:3.2}, $b\wedge a_n \in E$ for every $n \ge 1.$ We have $ b\wedge
a_1, b\wedge a_2, \ldots  \le a,b$.  Due to \cite[Prop 4.1]{Pul2}, $E$ satisfies $\sigma$-\RIP. Hence, we can find an element $b_0 \in E$ such
that $b\wedge a_1, b\wedge a_2, \ldots  \le b_0 \le a,b.$

\vspace{2mm} \noindent  {\it Claim 1.} $b \longleftrightarrow a.$

It is clear that $a = (a-b_0) +b_0$ and $b = (b-b_0) + b_0.$ In addition, we have $a_n \le (b -(b \wedge a_n))' \le
(b-b_0)'$ for every integer  $n \ge 1,$ so that $a \le (b-b_0)'$ which gives $a +(b-b_0) \in E$ and $a +(b-b_0)=
(a-b_0) +(b-b_0) +b_0 \in E.$

\vspace{2mm}\noindent  {\it Claim 2.} $\bigwedge_n (b -(b\wedge a_n))  = b - b_0.$

We have $b -(b\wedge a_n) \ge b-b_0$ for any integer $n \ge 1.$ Let
$d \le b-(b\wedge a_n) $  for $n\ge 1.$ By (viii) of Proposition \ref{pr:3.1}, $d
\le b -(b \wedge a_n) = (b \vee a_n) - a_n$. Since $b\wedge a_n, b\vee a_n$ exist in $E,$ by the proof of Claim 1, we have
$d+a_n \le b \vee a_n \le (b-b_0) +a,$ so that $a_n \le
((b-b_0)+a) - d$ and $a \le ((b-b_0)+a)-d$ which gives $d+a \le
(b-b_0) +a$ and $d \le b-b_0.$

\vspace{2mm}\noindent  {\it Claim 3.} $(b - b_0) \wedge (a-b_0) =
0.$

Assume $z \le b - b_0$ and $z \le a-a_0$. Then $z + b_0 \le b$ and
$z +b_0 \le a.$ Moreover, $b \wedge a_n \le z +b_0 \le b,a$ for
$n\ge 1.$ Applying Claims 1--2 for the element $z+b_0$, we have $b-b_0
= b-(z+b_0)$, i.e., $z =0.$

\vspace{2mm}\noindent  {\it Claim 4.} $a \stackrel{\mbox{\rm
c}}{\longleftrightarrow} b$ and $a \vee b, a\wedge b \in E.$

It follows from Claim 3 and (vi) of Proposition \ref{pr:3.1}.

\vspace{2mm}\noindent  {\it Claim 5.} $b \wedge (\bigvee_n a_n) =
\bigvee_n (b\wedge a_n) .$

It is clear that $b \wedge a \ge b\wedge a_n$ for each $n \ge 1.$ Assume
$b\wedge a_n \le y$ for some $y \in E$ and each integer $n\ge 1.$ Then $b\wedge a_1, b\wedge
a_2, \ldots \le y, b_0$ so that there exists an element $y_0 \in E$ such
that $b\wedge a_n \le y_0 \le y,b_0$ for $n\ge 1.$ Then $b - y_0
\le b-(b\wedge a_n).$ By Claim 2, we have $b-y_0 \le
\bigwedge_n (b-(b\wedge a_n)) = b- b_0$, so that $b_0 \le
y_0 \le y$ which finishes the proof.
\end{proof}

We note that a maximal set of mutually strongly compatible
elements of $E$ is said to be a {\it block} of $E$. For example,
if $E$ is an MV-algebra, then it is a unique block of $E$.

According to \cite{CIJTP}, we say that an effect algebra $E$ satisfies the {\it
Difference-Meet Property} (DMP for short),  if $x\le y$, $x\wedge
z \in E$ and $y \wedge z \in E$ imply $(y-x)\wedge z \in E.$ For
example, (1) every lattice-ordered effect algebra satisfies the
difference-met property, (ii) $\mathcal L(H)$ is a monotone $\sigma$-complete effect algebra (in fact a complete lattice) with RIP and DMP, but RDP fails for it.

The following example is from \cite{CIJTP}. We notice that a poset $E$ is an {\it antilattice} if $a\wedge b$ ($a \vee b$) exists in $E$ iff $a$ and $b$ are comparable.

\begin{example}\label{ex:3.3}
Let $G$ be the additive group
$\mathbb R^2$ with the positive cone of all $(x,y)$ such that
either $x=y=0$ or $x>0$ and $y>0$. Then $u=(1,1)$ is a strong unit
for $G$. The effect algebra $E=\Gamma(G,u)$ is an antilattice
having  {\rm RDP} but $E$ is not  a lattice and {\rm DMP} fails.
\end{example}

\begin{proof}
Let $x = (0.2, 0.3) \le y=(0.3, 0.5)$ and $z =(0.01, 0.25).$ Then $x\wedge z$ and $y\wedge z$ are defined in $E$ but $(y-x)\wedge z$ fails in $E.$
\end{proof}

The following property was proved in \cite[Prop 3.1]{CIJTP}:

\begin{proposition}\label{pr:3.4}
Let an effect algebra $E$ satisfy \RDP and {\rm DMP}. Then:
\begin{enumerate}
\item[{\rm (i)}] If $a \stackrel{\mbox{\rm c}}{\longleftrightarrow} b$,
then $a \stackrel{\mbox{\rm c}}{\longleftrightarrow} b'.$

\item[{\rm (ii)}] If $a \stackrel{\mbox{\rm
c}}{\longleftrightarrow} b$, $a \stackrel{\mbox{\rm
c}}{\longleftrightarrow} c$ and $b \le c$, then $a
\stackrel{\mbox{\rm c}}{\longleftrightarrow} (c-b).$
\end{enumerate}
\end{proposition}

\begin{theorem}\label{th:3.5}
Every block of a monotone $\sigma$-complete effect algebra $E$
with \RIP and {\rm DMP} is an effect subalgebra of $E$
which is a $\sigma$-complete MV-algebra. Moreover, any such  effect algebra is a
set-theoretical union of its blocks.
\end{theorem}

\begin{proof} Let $M$ be a block of $E$. Then $0,1 \in M$. Assume
$a,b,c \in M$. By Proposition \ref{pr:3.4}, $a \stackrel{\mbox{\rm
c}}{\longleftrightarrow} b'$ so that $b ' \in M.$ If $b+c \in E$,
we have by Proposition \ref{pr:3.4}, $a \stackrel{\mbox{\rm
c}}{\longleftrightarrow} (b'-c)$, consequently $a
\stackrel{\mbox{\rm c}}{\longleftrightarrow}(b'-c)'= b+c$ which
proves that $M$ is an effect subalgebra of $E$. Moreover, if $x,y
\in M$, by Propositions \ref{pr:3.2}, $x\vee y \in M,$ consequently,
$x\wedge y =(x'\vee y')' \in M.$

In view of (vi)--(viii) of  Proposition \ref{pr:3.1}, we have $(x\vee y) -y =
x - (x\wedge y)$ for all $x,y \in M$ which is a necessary and
sufficient condition  in order $M$ to be an MV-algebra.

Let $\{x_n\}$ be a sequence of elements of $M.$ Then $a_n:=x_1\vee\cdots \vee x_n \in M$ and by the hypothesis, $x =\bigvee_n x_n$ is defined in $E,$ and $y \scom a_n$ for every $n\ge 1.$ By Proposition \ref{pr:3.2}, $x \scom y$ for every $y \in M.$ Hence, $x\in M$ which proves $M$ is a $\sigma$-complete MV-algebra.

If $A$ is any subset of mutually strongly compatible elements of
$E$, due to Zorn's lemma there exists a block of $E$ containing
$A$. In particular, by Proposition \ref{pr:3.1}, the set $A =
\{0,a,a',1\}$ $(a \in E)$ is a set of mutually strongly compatible
elements of $E$. This proves that $E$ is a set-theoretical union
of its blocks.
\end{proof}

\begin{corollary}\label{co:3.6}
Every $\sigma$-complete effect algebra $E$ can be covered by blocks and every block of $E$ is a $\sigma$-complete MV-algebra.
\end{corollary}

\begin{proof}
Our effect algebra is in view of hypothesis a monotone $\sigma$-complete effect algebra with RIP and DMP. The desired result follows from Theorem \ref{th:3.5}.
\end{proof}

According to \cite{Jen}, we say that a finite non-empty set $M=\{a_1,\ldots,a_n\}$ of an effect algebra $E$ is {\it jointly compatible} if there is a finite sequence of summable elements $(c_1,\ldots,c_k)$ such that, for every $a_i,$ $i=1,\ldots,n,$ there is a finite subsequence $(c_{k_1},\ldots, c_{k_i})$ of $(c_1,\ldots,c_k)$ with $a_i = c_{k_1}+\cdots+c_{k_i}.$ An arbitrary subset $M$ of $E$ is {\it jointly compatible} if any its finite subset is jointly compatible.

We say that an arbitrary subset $M$ of $E$ is {\it internally compatible} if, for any finite subset $M_F$ of $M,$ there is a summable sequence $(c_1,\ldots,c_k)$  of elements of $M$ such that every element of $M_F$ can be expressed as a sum of some elements from $\{c_1,\ldots,c_k\}.$ We note that a two-element set $\{a,b\}$ is jointly compatible iff $a \leftrightarrow b.$

\begin{proposition}\label{pr:3.7}
Let $x$ be a  homomorphism from the Borel $\sigma$-algebra $\mathcal B(\mathbb R)$ into an effect algebra $E.$ Then the range of $x$, $\mathcal R(x):=\{x(A)\colon A \in \mathcal B(\mathbb R)\},$ is an internally compatible subset of $E.$

Conversely, let $\{a_1,\ldots,a_k\}$ be a finite subset of jointly compatible elements of $E.$ Then there is a homomorphism $x:\mathcal B(\mathbb R) \to E$ such that $\{a_1,\ldots,a_n\} \subseteq \mathcal R(x).$
\end{proposition}

\begin{proof}
(1) Let $x(A_1),\ldots,x(A_n)$ with $A_1,\ldots, A_n \in \mathcal B(\mathbb R)$ be given. We define Borel sets $A(i_1,\ldots,i_n) = A^{i_1}_1\cap \cdots \cap A^{i_n}_n,$ where $i_1,\ldots,i_n \in \{0,1\}$ and $A^1=A$ and $A^0= \mathbb R \setminus A.$ Then every $x(A_i)$ can be expressed as a finite sum of elements from the finite system $(x(A(i_1,\ldots,i_n))\colon i_1,\ldots, i_n \in \{0,1\}).$

(2)  Suppose that $a_1,\ldots,a_n$ are jointly compatible elements of $E$. There is a finite summable sequence $(c_1,\ldots,c_k)$ such that every $a_i$ is a sum of some elements the sequence $(c_1,\ldots,c_n).$ We set $c_0=(c_1+\cdots+c_k)'$ and define a homomorphism $x:\mathcal B(\mathbb R)\to E$ as follows

$$
x(A):= \sum\{c_i\colon i \in A\cap \{0,1,\ldots,k\}\},\quad A \in \mathcal B(\mathbb R).
$$
Then $\mathcal R(x)$ contains the elements $\{a_1,\ldots,a_n\}.$
\end{proof}

A subset $M$ of an effect algebra $E$ is said to be (i) an {\it ic-block} if $1 \in M$ and $M$ is a maximal internally compatible subset of $E,$ and (ii) an {\it RDP-block} if $M$ is a maximal effect subalgebra with RDP. Since the subset $\{0,1\}$ is both internally compatible as well as an effect subalgebra with RDP, by Zorn's Lemma, an ic-block as well as an RDP-block exist in any effect algebra.

By \cite[Thm 3.11]{Jen}, for a homogeneous effect algebra $E$ a subset $M$ of $E$ is an ic-block iff it is an RDP-block. In addition, $E$ can be covered by all its ic-blocks as well as by all its RDP-blocks.

\section{Representable Effect Algebras}

In this section we study representable monotone $\sigma$-complete effect algebras that is, those that are $\sigma$-homomorphic images of some effect-tribes or even of tribes.

We say that a monotone $\sigma$-complete effect algebra $E$ is (i) {\it representable} if there are a non-empty set $\Omega,$  an effect-tribe $\mathcal T \subseteq [0,1]^\Omega$ and a $\sigma$-homomorphism $h$ from $\mathcal T$ onto $E,$
(ii) {\it weakly representable} if there is a system of monotone $\sigma$-complete effect subalgebras $(E_t \colon t \in T)$ of $E$  such that every $E_t$ is representable and $\bigcup_{t \in T} E_t=E,$ (iii) {\it MV-weakly representable} if there is a system of $\sigma$-complete MV-algebras $(E_t\colon t \in T)$ such that $E_t \subseteq E$ for each $t \in T$ and $\bigcup_{t \in T}E_t=E,$ and (iv) {\it RDP-weakly representable} if there is a system of $\sigma$-complete effect algebras with RDP $(E_t\colon t \in T)$ such that $E_t \subseteq E$ for each $t \in T$ and $\bigcup_{t \in T}E_t=E.$

It is clear that (1) if a monotone $\sigma$-complete effect algebra $E$ is representable (weakly representable), so is any $\sigma$-homomorphic image of $E.$ (2) If  $(E_t: t \in T)$ is a system of monotone $\sigma$-complete effect algebras such that each $E_t$ is representable (weakly representable), so is the horizontal sum of $(E_t\colon t\in T).$ (3)  If $(E_t \colon t \in T)$ is a system of monotone $\sigma$-complete effect algebras with RDP, then the direct product effect algebra $\prod_{t \in T}E_t$ with effect algebra operations defined by coordinates has RDP and is monotone $\sigma$-complete, consequently, it is representable. (4)  If $E$ is MV-weakly representable, then it is RDP-weakly representable.

We say that a system of states $\mathcal S$ of an effect algebra $E$ is {\it order-determining} if $a\le b$ for $a,b \in E$ iff $s(a)\le s(b)$ for any $s \in \mathcal S.$

\begin{theorem}\label{th:4.1}
Any of the following monotone $\sigma$-complete effect algebras is representable:

\begin{enumerate}
\item[{\rm (i)}] A $\sigma$-complete MV-algebra.

\item[{\rm (ii)}] A monotone $\sigma$-complete effect algebra with \RDP.

\item[{\rm (iii)}] A monotone $\sigma$-complete effect algebra with an order determining system of $\sigma$-additive states.

\item[{\rm (iv)}] $\mathcal E(H)$ and $\mathcal L(H).$

\end{enumerate}
\end{theorem}

\begin{proof}
(i) If $E$ is a $\sigma$-complete MV-algebra, by \cite{Mun, LoomD}, $E$ is always a $\sigma$-homomorphic image of a tribe which is a special kind of an effect-tribe.

(ii) Due to \cite{BCD}, every monotone $\sigma$-complete effect algebra with RDP is a $\sigma$-homomorphic image of an effect-tribe with RDP.

(iii)  Let $\mathcal S$ be an order-determining system of $\sigma$-additive states on $E.$ Given a $\sigma$-additive state $s\in \mathcal S,$ let $\hat s:E \to [0,1]$ be defined by $\hat s(a):= s(a),$ $a \in E.$ Then $\widehat {\mathcal S}:=\{\hat s\colon s \in \mathcal S\}$ can be naturally organized into an effect-tribe such that $\widehat {\mathcal S}$ is isomorphic to $E.$

(iv)  Let  $\phi \in H$ be a unit vector in the Hilbert space $H.$ It defines a $\sigma$-additive state $s_\phi$ on $\mathcal E(H)$ by $s_\phi(A):=(A\phi,\phi),$ $A \in \mathcal E(H).$ Then $\{s_\phi\colon \phi \in H\}$ is an order-determining system of $\sigma$-additive states on $\mathcal E(H).$ The desired result follows from (iii).

The case $\mathcal L(H)$ is identical.
\end{proof}

Let $M$ be a non-empty subset of a monotone $\sigma$-complete effect algebra $E$. We say that $M$ is {\it strongly monotone $\sigma$-complete} if, for any non-decreasing sequence $\{a_n\}$ of elements of $M$ such that $\bigvee_n a_n= a$ in $E,$ then $a \in M.$

\begin{theorem}\label{th:4.2}
Any of the following monotone $\sigma$-complete effect algebras is weakly representable:

\begin{enumerate}
\item[{\rm (i)}] A monotone $\sigma$-complete effect algebra with \RIP and {\rm DMP}.

\item[{\rm (ii)}] A $\sigma$-complete effect algebra.

\item[{\rm (iii)}] A $\sigma$-orthocomplete orthomodular poset.

\item[{\rm (iv)}] A monotone $\sigma$-complete homogeneous effect algebra such that every RDP-block is strongly monotone $\sigma$-complete.

\end{enumerate}
\end{theorem}

\begin{proof}
(i) By Theorem \ref{th:3.5}, $E$ can be covered by a system of blocks $(E_t: t\in T)$ and each block $E_t$ is a $\sigma$-complete MV-algebra. By (i) of Theorem \ref{th:4.1}, each block is representable.

(ii) It follows from Corollary \ref{co:3.6} and (i).

(iii) If $E$ is a $\sigma$-orthocomplete orthomodular poset, by \cite[Sec 2.1.3]{Dvu}, $E$ can be covered by a system of Boolean $\sigma$-algebras. By the classical Loomis-Sikorski Theorem,  \cite[Thm 29.1]{Sik}, any Boolean $\sigma$-algebra is a $\sigma$-homomorphic image of a $\sigma$-algebra $B$ of subsets of a set $\Omega \ne \emptyset.$ If we take now the set of all characteristic functions of the sets from $B,$ we see that $B$ is representable. In addition, every $\sigma$-orthocomplete effect algebra can be covered by blocks.

(iv)  By \cite[Thm 3.11]{Jen}, $E$ can be covered by a system of RDP-blocks $(M_t:t \in T)$ and every such  RDP-block $M_t$ is an effect subalgebra of $E$ and $E_t$ satisfies RDP. Assuming $M_t$ is strongly monotone $\sigma$-complete, we see that $E$ is weakly representable.
\end{proof}

\begin{theorem}\label{th:4.3}
The following monotone $\sigma$-complete effect algebras are MV-weakly representable:
\begin{enumerate}
\item[{\rm (i)}] $\mathcal E(H)$ if $H$ is a separable complex Hilbert space.

\item[{\rm (ii)}] A monotone $\sigma$-complete effect algebra with \RIP and {\rm DMP}.

\item[{\rm (iii)}] A $\sigma$-complete effect algebra.

\item[{\rm (iv)}] A $\sigma$-orthocomplete orthomodular poset.
\end{enumerate}
\end{theorem}

\begin{proof}
(i) According to von Neumann \cite{vNe}, every set of mutually commuting elements in $\mathcal E(H)$ can be embedded into a maximal Abelian von Neumann algebra $\mathcal A$, and every element of $\mathcal A$   is a real-valued function of some element $a$ in $\mathcal A.$ This can be by \cite{Pul1, Var} transformed into an MV-algebra (as functions of $a$) which is even  a $\sigma$-complete MV-algebra. Consequently, $\mathcal E(H)$ can be covered by a system of $\sigma$-complete MV-algebras.

(ii) It follows from (i) of Theorem \ref{th:4.1}.

(iii) It follows from (ii) of Theorem \ref{th:4.1}.

(iv)  By the proof of (iii) of Theorem \ref{th:4.1}, $E$ can be covered by a system of Boolean $\sigma$-algebras, and every Boolean algebra is in its own right an MV-algebra.
\end{proof}

\section{Observables and Representability of Effect Algebras}

This section studies observables and we show when partial information about an observable on all intervals of the form $(-\infty,t)$ implies the whole information about the observable.

Let $E$ be a monotone $\sigma$-complete effect algebra and let $\mathcal B(\mathbb R)$ be the Borel $\sigma$-algebra of the real line $\mathbb R.$ A mapping $x: \mathcal B(\mathbb R) \to E$ is said to be an {\it observable} on $E$ if (i) $x(\mathbb R)=1,$ (ii) if $A$ and $B$ are mutually disjoint Borel sets of $\mathbb R$, then $x(A \cup B)=x(A)+x(B),$ where $+$ is the partial addition on $E,$ and (iii) if $\{A_i\}$ is a sequence of Borel sets such that $A_i \subseteq A_{i+1}$ for every $i$ and $A= \bigcup_i A_i,$ then $x(A) = \bigvee_i x(A_i).$ In other words, an observable is a $\sigma$-homomorphism of monotone $\sigma$-complete effect algebras.

We notice that for all $A,B \in \mathcal B(\mathbb R)$, we have (i) $x(\mathbb R \setminus A) = x(A)',$ (ii) $x(\emptyset) = 0,$ (iii) if $A \subseteq B,$ then  $x(A)\le x(B)$   and $x(B \setminus A) = x(B) - x(A)$  (iv) if $B_i \supseteq B_{i+1}$ and $B = \bigcap_i B_i$, then $x(B) = \bigwedge _i x(B_i).$

We denote by $\mathcal R(x):=\{x(A): A \in \mathcal B(\mathbb R)\},$ the {\it range} of $x.$ Then the range is not necessarily  an effect subalgebra as well as not an MV-subalgebra of $M,$ see \cite{DvKu}.

\begin{proposition}\label{pr:5.1}
Let $x$ be an observable on a monotone $\sigma$-complete effect algebra $E$. Given a real number $t \in \mathbb R,$ we put

$$ x_t := x((-\infty, t)). \eqno(5.1)
$$
Then

$$ x_t \le x_s \quad {\rm if} \ t < s, \eqno (5.2)$$

$$\bigwedge_t x_t = 0,\quad \bigvee_t x_t =1, \eqno(5.3)
$$
and
$$ \bigvee_{t<s}x_t = x_s, \ s \in \mathbb R. \eqno(5.4)
$$
\end{proposition}

\begin{proof}
The statement is straightforward.
\end{proof}

The aim of the remainder of the paper is to exhibit situations when a system of elements $\{x_t: t \in \mathbb R\}$ of a monotone $\sigma$-complete effect algebra satisfying (5.2)--(5.4) imply existence of an observable $x$ on $E$ such that (5.1) holds for any $t \in \mathbb R.$ Many important situations were positively solved in \cite{DvKu} for monotone $\sigma$-complete effect algebras with RDP, $\mathcal E(H),$ $\sigma$-MV-algebras, $\sigma$-lattice effect algebras. Now we concentrate to monotone $\sigma$-complete effect algebras satisfying RIP and DMP, and to some kind of homogeneous effect algebras.

\begin{theorem}\label{th:5.2}
Let $E$ be a monotone $\sigma$-complete effect algebra $E$ with {\rm RDP} and {\rm DMP}. Let $(x_t: t \in \mathbb R)$ be a system of elements of $E$ satisfying {\rm (5.2)--(5.4)}, then there is a unique observable $x$ on $E$ for which $(4.1)$ holds for any $t \in \mathbb R.$
\end{theorem}

\begin{proof}
It is evident that $(x_t: t \in \mathbb R)$ is a system of mutually strongly compatible elements of $E.$ According to Theorem \ref{th:3.5}, there is a block $M$ of $E$ containing all elements of $\{x_t:t \in \mathbb R\}.$ In addition, this block is a $\sigma$-complete MV-algebra.
To prove the statement, we outline only the main steps of the proof; for more details see \cite[Thm 3.2, Thm 3.9]{DvKu}. According to Theorem \ref{th:4.1}(i), $M$ is representable, that is, there are an effect-tribe $\mathcal T \subseteq [0,1]^\Omega,$ $\Omega \ne \emptyset,$ a $\sigma$-homomorphism $h$ from $\mathcal T$ on $M.$

Let $r_1,r_2,\ldots$ be any enumeration of the set of rational numbers.   Given $r_n,$ let $a_n$ be a function from the tribe $\mathcal T$ such that $h(a_n)=x_{r_n}$ for any $n\ge 1.$ We are stating that it is possible to find such a sequence of functions $\{b_n\}$ from $\mathcal T$ such that $
h(b_n)=x_{r_n}$ for any $n\ge 1$ and $b_n \le b_m$ whenever $r_n < r_m.$ Indeed, if $n=1$, we set $b_1 = a_1.$ By mathematical induction suppose that we have find $b_1,\ldots,b_n$ such that $h(b_i)=x_{r_i},$ and $b_i \le b_j$ whenever $r_i < r_j$ for $i,j=1,\ldots, n.$ Let $j_1,\ldots,j_n$ be a permutation of $1,\ldots,n$ such that $r_{j_1}<\cdots<r_{j_n}.$  For $r_{n+1}$ we have three possibilities (i) $r_{n+1}< r_{j_1},$ (ii) there exists $k =1,\ldots,n-1$ such that $r_{j_k} < r_{n+1} < r_{j_{k+1}},$ or (iii) $r_{j_n} < r_{n+1}.$  Applying \cite[Lem 2.1]{DvKu}, we can find $b_{n+1}\in \mathcal T,$ $h(b_{n+1})=r_{n+1},$ such that for all $i,j =1,\ldots, n+1,$ $b_i \le b_j$ whenever $r_i < r_j.$

Due to the density of rational numbers, for any $t \in \mathbb R,$ we can find an element $b_t \in \mathcal T$ such that $h(b_t)=r_t.$

Fix $\omega \in \Omega$ and define $F_\omega(t):=b_t(\omega),$ $t \in \mathbb R.$ Then $F_\omega$ is a non-decreasing, left continuous functions with  $\lim_{t \to -\infty} F_\omega(t)=0$  and $\lim_{t \to \infty} F_\omega(t)=1.$ By \cite[Thm 43.2]{Hal}, $F_\omega$ is a distribution function on $\mathbb R$ corresponding to a unique probability measure $P_\omega$ on $\mathcal B(\mathbb R),$ that is, $P_\omega((-\infty,t))=F_\omega(t)$ for every $t \in \mathbb R.$ Define now a mapping $\xi: \mathcal B(\mathbb R) \to [0,1]^\Omega$ by $\xi(E)(\omega)=P_\omega(E),$ $E \in \mathcal B(\mathbb R)$, $\omega \in \Omega.$  It is possible to show that $\{E \in \mathcal B(\mathbb R)\colon \xi(E) \in \mathcal T\}= \mathcal B(\mathbb R).$

Therefore, $x:= h \circ \xi$ is an observable on $M$ such that $x((-\infty, t))=x_t$ for any $t \in \mathbb R.$ To prove the uniqueness of $x$, it is possible to show that $\{E \in \mathcal B(\mathbb R) \colon x(E)=\xi(E)\}=\mathcal B(\mathbb R).$
\end{proof}

\begin{theorem}\label{th:5.3}
Let $E$ be a monotone $\sigma$-complete homogeneous effect algebra $E$ such that any \RDP-block of $E$ is strongly monotone $\sigma$-complete. Let $(x_t: t \in \mathbb R)$ be a system of elements of $E$ satisfying {\rm (5.2)--(5.4)}, then there is a unique observable $x$ on $E$ for which $(4.1)$ holds for any $t \in \mathbb R.$
\end{theorem}

\begin{proof}
Let $x_{-\infty}:=0$ and $x_\infty :=1.$
The set $M_0:=\{x_t-x_s \colon s\le t,\ s,t \in \mathbb R \cup\{-\infty\}\cup \{\infty\}\}$ is internally compatible, it contains all $x_t$'s, and $1 \in M_0.$ By \cite[Thm 3.11]{Jen}, there is an ic-block $M$ of $E$ that is also an RDP-block of $E$ such that $M_0 \subseteq M.$ By the assumption, $M$ is a strongly monotone $\sigma$-complete effect algebra with RDP. By (ii) of Theorem \ref{th:4.1}, $M$ is representable.  Using the methods of the proof of Theorem \ref{th:5.2}, see also \cite[Thm 3.9]{DvKu}, we can find a unique observable $x$ of $M,$ consequently of $E$ such that $x_t = x((-\infty,t)),$ $t \in \mathbb R.$
\end{proof}

We note that Theorem \ref{th:5.3} describes a special type of RDP-weakly representable monotone $\sigma$-complete effect algebras.






\begin{thebibliography}{DvVe2}

\bibitem[BCD]{BCD}
D. Buhagiar,  E. Chetcuti,  A. Dvure\v censkij,    {\it
Loomis-Sikorski representation of  monotone $\sigma$-complete effect
algebras}, Fuzzy Sets and Systems {\bf 157} (2006), 683--690.

\bibitem[Cha]{Cha}
C.C. Chang, {\it  Algebraic analysis of many-valued logics}, Trans.
Amer. Math. Soc. {\bf 88} (1958), 467--490.


\bibitem[Dvu]{Dvu}
A. Dvure\v censkij,    {\it  ``Gleason's Theorem and Its
Applications"}, Kluwer Academic Publisher,
Dordrecht/Boston/London, 1993, 325+xv pp.



\bibitem[Dvu1]{LoomD}
A. Dvure\v censkij,    {\it Loomis--Sikorski
theorem  for $\sigma$-complete MV-algebras and $\ell$-groups},
J. Austral. Math. Soc. Ser. A {\bf 68}  (2000), 261--277.

\bibitem[Dvu2]{CIJTP}
A. Dvure\v censkij,    {\it  On effect algebras which can be
covered by MV-algebras}, Inter. J. Theor. Phys. {\bf 41} (2002),
221--229.

\bibitem[DvKu]{DvKu}
A. Dvure\v censkij, M. Kukov\'a,   {\it Observables on quantum structures}, Inf. Sci.
DOI:10.1016/j.ins.2013.09.014

\bibitem[DvPu]{DvPu} A. Dvure\v censkij, S. Pulmannov\'a,   {\it ``New
Trends in Quantum Structures"}, Kluwer Academic Publ.,
Dordrecht, Ister Science, Bratislava, 2000, 541 + xvi pp.


\bibitem[FoBe]{FoBe}
D.J. Foulis, M.K. Bennett,
{\it  Effect algebras and unsharp quantum logics}, Found. Phys. {\bf
24} (1994), 1325--1346.


\bibitem[Go]{Goo}
 K.R. Goodearl,
{\it ``Partially Ordered Abelian Groups with Interpolation"},
 Math. Surveys and Monographs No. 20, Amer. Math. Soc.,
 Providence, Rhode Island, 1986.

\bibitem[Hal]{Hal}
P.R. Halmos, {\it ``Measure Theory"}, Springer-Verlag, Berlin, 1974.


\bibitem[Jen]{Jen}
G. Jen\v{c}a, {\it Blocks of homogeneous effect algebras}, Bull. Australian
Math. Soc. {\bf 64} (2001), 81--98.

\bibitem[Mun]{Mun}
D. Mundici, {\it   Tensor products and the Loomis--Sikorski theorem
for MV-algebras}, Advan. Appl. Math. {\bf 22} (1999),   227--248.

\bibitem[NiPa]{NiPa}
J. Niederle, J. Paseka, {\it Homogeneous orthocomplete effect algebras are covered by MV-algebras}, Fuzzy Sets and Systems {\bf 210} (2013), 89--101.

\bibitem[Pul1]{Pul1}
S. Pulmannov\'a,   {\it Compatibility and decomposition of effects}, J. Math. Phys. {\bf 43} (2002), 2817--2830.

\bibitem[Pul2]{Pul2}
S. Pulmannov\'a, {\it Blocks in homogeneous effect algebras and MV-algebras},
Math. Slovaca {\bf 53} (2003), 525--539.

\bibitem[Rie]{Rie}
Z. Rie\v canov\'a, {\it A generalization of blocks for lattice
effect algebras}, Inter. J. Theoret. Phys. {\bf 39} (2000),
231--237.

\bibitem[Sik]{Sik}
R. Sikorski, {\it ``Boolean Algebras"}, Springer--Verlag, Berlin, Heidelberg, New York, 1964.

\bibitem[Var]{Var} V.S. Varadarajan,
{\it ``Geometry of Quantum Theory", Vol. 1},
 van Nostrand,
  Princeton, New Jersey, 1968.

\bibitem[vNe]{vNe}
J. von Neumann, {\it ``Mathematical Foundations of Quantum Mechanics"}, Princeton Univ, Press, 1955.

\end{thebibliography}
\end{document}